\documentclass[journal, towcolumn, 10pt]{IEEEtran}
\usepackage{graphicx}
\usepackage{algorithm}
\usepackage{algorithmicx}
\usepackage{algpseudocode}
\usepackage{cite}
\usepackage{color}
\usepackage{amsmath,amssymb,amsfonts,amsthm}
\usepackage{bm}
\usepackage{setspace}
\usepackage{multirow}
\usepackage{mathtools}

\newtheorem{thm} {Theorem}

\begin{document}
\title{Node Removal Vulnerability of the Largest Component of a Network}

\author{Pin-Yu~Chen and Alfred O. Hero III,~\emph{Fellow},~\emph{IEEE}%<-this % stops a space
\\ Department of Electrical Engineering and Computer Science, University of Michigan, Ann Arbor, USA
\\Email : \{pinyu,hero\}@umich.edu
\thanks{This work has been partially supported by a Department of EECS Graduate Fellowship to the first author and by the Army Research Office (ARO), grant number W911NF-09-1-0310.}
}

\maketitle
%\setstretch{1.9}
\thispagestyle{empty}
\begin{abstract}
The connectivity structure of a network can be very sensitive to removal of certain nodes in the network. In this paper, we study the sensitivity of the largest component size to node removals.
We prove that minimizing the largest component size is equivalent to solving a matrix one-norm minimization problem whose column vectors are orthogonal and sparse and they form a basis of the null space of the associated graph Laplacian matrix. A greedy node removal algorithm is then proposed based on the matrix one-norm minimization. In comparison with other node centralities such as node degree and betweenness, experimental results on US power grid dataset validate the effectiveness of the proposed approach in terms of reduction of the largest component size with relatively few node removals.
\end{abstract}

\begin{IEEEkeywords}
graph Laplacian, greedy node removal, network robustness, spectral graph theory, topological vulnerability
\end{IEEEkeywords}
%\IEEEpeerreviewmaketitle

\section{Introduction}
\label{sec_Intro}
Networks are vulnerable to selective node removals and even a few such removals can severely disrupt their operation \cite{Albert00}. The sensitivity of the size of the largest component
to node removals is one of the most important topological vulnerability measures in network science \cite{Lewis08}, as it is closed related to the functionality, robustness and fragility \cite{Callaway00,Xiao08,Sole08,CPY12}.
Despite its wide range of interest, little is known on how one might most efficiently disrupt a network given a fixed number of node removals and how to identify the most vulnerable nodes.

A phase transition occurs when the fraction of removed nodes exceeds certain critical value, and the largest component vanishes into several small components. Under uncorrelated random graph assumptions,
Cohen \textit{et. al.} \cite{Cohen01} use degree distributions to evaluate the critical value for this phase transition based on node degree removals.
However, it has been shown in \cite{Alderson05,Moreira09} that the uncorrelated graph assumption is a poor fit to some real world networks.

Another commonly adopted node centrality for studying network connectivity is betweenness centrality \cite{Freeman77}.
The betweenness of a node $v$ is defined as
\begin{align}
\label{eqn_betweenness}
\sigma(v)=\sum_{s \neq v \neq t} \frac{\sigma_{st}(v)}{\sigma_{st}},
\end{align}
where $\sigma_{st}$ is the total number of shortest paths from node $s$ to node $t$
and $\sigma_{st}(v)$ is the number of those paths that pass through $v$.
Roughly speaking, a node is regarded as more important if it is bypassed by more shortest paths in the network.
Holme \textit{et. al.} \cite{Holme02} have shown that greedy node removals can be made more harmful by iteratively removing the node with the highest node degree. They also empirically verify that betweenness based node removal is more effective than degree based node removal in terms of minimizing largest component size. In this paper, we propose a new network robustness measure that is directly related to the largest component size. When used to minimize the largest component size, our proposed measure outperforms the degree and betweenness centrality.

The proposed measure is based on the graph Laplacians \cite{Mohar91}. The graph Laplacian has been widely used to characterize graph connectivity. We establish a link between the graph Laplacian and the size of the largest component. Specifically, we show that minimizing largest component size is equivalent to finding a set of sparse orthogonal vectors that span the null space of the associated graph Laplacian matrix. The equivalence is exact and it imposes no restrictive assumptions, e.g. uncorrelatedness, small world network structure, scale-free degree distribution, etc.
Based on the formulation, a spectral graph cut based greedy node removal procedure is proposed to identify the most vulnerable nodes.

To illustrate our proposed method, we use United States power grid topology.
Comparing with strategies based on node degree and betweenness, our proposed graph Laplacian node removal approach leads to a selection of nodes whose removal significantly increases the rate of reduction of the largest component size. This results in a useful measure of network sensitivity that can be used to asses network vulnerability to node removals.

\section{Properties of Graph Laplacians}
\label{sec_Laplacian}
Consider an unweighted and undirected network containing no self loops or multiple edges, the corresponding network graph can be denoted by a simple graph $G=(V,E)$ with node set $V$ and edge set $E=\{(u,v):u,v \in V\}$, where $|V|=n$ and $|E|=m$ are the number of nodes and edges in the graph, respectively.
The adjacency matrix $A$ of $G$ is a binary symmetric $n$-by-$n$ matrix, where $A_{ij}=1$ if $(i,j)\in E$ and otherwise $A_{ij}=0$. Let $d_i$ denote the number of edges incident to node $i$, the degree matrix $D=\textnormal{diag}(d_1,d_2,\ldots,d_n)$ is a diagonal matrix with its entry $D_{ii}=d_i$.
The graph Laplacian matrix $L$ of $G$ is defined as $L=D-A$, and $L$ can be decomposed by the outer product of an $n$-by-$m$ signed incidence matrix $B$ such that $L=BB^T$. For any $e=(v,w)\in E$, $v<w$, $B_{v,e}=1$ and $B_{w,e}=-1$, otherwise $B_{v,e}=0$. Therefore, $L$ is a symmetric and positive semidefinite (PSD) matrix. Similarly, the signless graph Laplacian matrix $Q$ is defined as $Q=D+A$ \cite{Cvetkovic07}, where $Q$ is PSD, symmetric and its incidence matrix is the signless incidence matrix of $G$.

Let $\lambda_i(L)$ be the $i$-th largest eigenvalue of $L$ and $\textbf{1}$ denote the all one vector. Since $L\textbf{1}=(D-A)\textbf{1}=0$, $\textbf{1}$ is always in the null space of $L$ and the smallest eigenvalue $\lambda_n(L)=0$. Furthermore, the multiplicity of zero eigenvalues is equal to the number of components (including the isolated nodes) of a simple graph \cite{Fiedler73}.
The nuclear norm of $L$ is associated with the number of edges in $G$ as $\|L\|_\ast=\sum_{i=1}^n \lambda_i(L)=2|E|$. Therefore, the highest-degree-first node removal strategy is in fact a greedy nuclear norm minimization heuristic.
The aforementioned properties of graph Laplacians will be useful in analyzing the network robustness to node removals.

\section{Network Robustness Measure Formulation}
\label{sec_measure}
Let $R \subset V$ denote the set of removed nodes from $G$ with $|R|=q$ and $G_R=(V_R,E_R)$ denote the remaining graph after removing the nodes in $R$ from $G$. When a node is removed, the edges attached to the node will be removed as well. The network vulnerability to node removals is evaluated in terms of the remaining largest component size after node removals, where we denote the size and the number of edges of the largest component of $G_R$ by $|V_R^{LC}|$ and $|E_R^{LC}|$, respectively.
In general, the optimal node removal set $R$ that minimizes $|V_R^{LC}|$ is not unique. Therefore, we propose a node removal approach that minimizes $|V_R^{LC}|$ and $|E_R^{LC}|$ simultaneously.
Specifically, we seek an $R=R^*$ that achieves
\begin{align}
\label{eqn_general_attack}
R^*=\min_{R \in F_q} |E^{LC}_R|,
\end{align}
where $F_q=\{R:R=\arg \min_{R^\prime,~|R^\prime|=q} |V^{LC}_{R^\prime}| \}$ is the solution space containing the feasible node removal sets that minimize the largest component size. In general, the computational complexity for solving this problem is of combinatorial order $\binom{n}{q}$. We use graph Laplacians to formulate the network metrics $|V_{R}^{LC}|$ and $|E_{R}^{LC}|$ and propose a greedy node removal approach to reduce computational complexity.

\subsection{Upper bound on the number of edges in the largest component}
For a given set of removed nodes $R$, let $L_R$ denote the graph Laplacian matrix of $G_R$ and $\lambda_i(L_R)$ denote the $i$-th largest eigenvalue of $L_R$. The following theorem gives an upper bound on the number of edges $|E_{R}^{LC}|$ in the largest component.
\begin{thm}
\label{Thm_link}
$|E_R^{LC}|$ is upper bounded by $\frac{1}{8}(n+1) \lambda_1(\widetilde{Q}_R)$, where
$\widetilde{Q}_R=\left[
\begin{smallmatrix}
  Q_R & d  \\
  d^T & 0
\end{smallmatrix} \right]$,
$d=A_R \textnormal{\textbf{1}}$ and $Q_R=D_R+A_R$ is the signless graph Laplacian matrix of $G_R$.
\end{thm}
\begin{proof}
Let $A_R$ be the adjacency matrix of $G_R$ and $s$ be an $n \times 1$ identification vector such that $s_i=1$ if $i \in V^{LC}_R$, otherwise $s_i=-1$. We have
\begin{align}
|E_R^{LC}|&=\frac{1}{2} \sum_{i,j \in V^{LC}_R}[A_R]_{ij}
=\frac{1}{8} \sum_{i,j}[A_R]_{ij}(1+s_i)(1+s_j) \nonumber
\end{align}
\begin{align}
&=\frac{1}{8} \sum_{i,j}[A_R]_{ij}+\frac{1}{4} \sum_{i,j}[A_R]_{ij}s_i
+\frac{1}{8} \sum_{i,j}[A_R]_{ij}s_is_j \nonumber \\ % two col
&=\frac{1}{8} \sum_{i}d_{i} s_i^2+\frac{1}{4} \sum_{i}d_{i}s_i+\frac{1}{8} \sum_{i,j}[A_R]_{ij}s_is_j
\nonumber\\
&=\frac{1}{8}\left[s^T(D_R+A_R)s+2d^Ts\right] = \frac{1}{8}\left[s^T Q_R s+2d^Ts\right]\nonumber\\
&= \frac{1}{8} {s^\prime}^T \widetilde{Q}_R s^\prime \leq \frac{1}{8}(n+1) \lambda_1(\widetilde{Q}_R)
\end{align}
by the Rayleigh quotient theorem \cite{HornMatrixAnalysis}, where $s^\prime=[s~1]^T$. The upper bound is attained if $s^\prime / \sqrt{n+1}$ is an eigenvector of $\lambda_1(\widetilde{Q}_R)$.
\end{proof}

\subsection{Largest component size}
\label{subsec_LC}
Let $\textnormal{null}(L)$ denote the null space of $L$ and define the sparsity of a vector to be the number of zero entries in the vector.
Next we express the largest component size via graph Laplacians.
\begin{thm}
\label{Thm_largest_component_size}
$|V_R^{LC}|=\| X\|_1=\max_i \|x_i\|_1$, where $x_i$ is the $i$-th column vector of binary matrix $X$. The columns of $X$ are orthogonal and they form the sparsest basis of \textnormal{null}$(L_R)$ among binary vectors.
\end{thm}
\begin{proof}
Let $r$ be the rank of $L_R$. We will prove that there exists an $n \times (n-r)$ binary matrix $X=[x_1~x_2 \ldots x_{n-r}]$ whose columns $\{x_i\}_{i=1}^{n-r}$ satisfy: 1) $\|x_i\|_1$ is the size of the $i$-th component of $G_R$; 2) they are orthogonal.
Assume $G_R$ consists of $K$ components. There exits a matrix permutation (relabeling) such that
\begin{align}
L_R=
\left[
  \begin{matrix}
       L_R^1 & 0     & 0      & 0  \\
       0     & L_R^2 & 0      & 0 \\
       0     & 0     & \ddots & 0   \\
       0     & 0     & 0      & L_R^K %\nonumber
  \end{matrix}
\right].
\end{align}
Associated with the $i$-th block matrix $L_R^i$ we define $x_i$ as an $n \times 1$ binary vector $x_i$ in \textnormal{null}$(L_R)$ having the form
$x_i=[0 \ldots 0~1 \ldots 1~0 \ldots 0]^T$, where the locations of the nonzero entries correspond to the indexes of the $i$-th block matrix. It is obvious that $\| x_i \|_1=\sum_{j=1}^{n}|x_{ij}|$ equals the size of the $i$-th component and the $\{x_i\}_{i=1}^{n-r}$ are mutually orthogonal. Furthermore, there exists no other binary matrix which is sparser than $X$ with column span equal to \textnormal{null}$(L_R)$. If there existed another binary matrix that were sparser than $X$, then it contradicts the fact the its column vectors characterize the component sizes of $G_R$. Therefore the largest component size of $G_R$ is $|V_R^{LC}|=\| X\|_1=\max_i \|x_i\|_1$.
\end{proof}

\subsection{Greedy basis search algorithm for constructing $X$}
It has been proven in \textbf{Theorem 2.1} of \cite{Coleman86} that a matrix $X$ is a sparsest basis for a finite dimensional linear subspace if and only if it can be constructed by greedy basis search. This result will allow us to solve for the solution $X$ in \textbf{Theorem \ref{Thm_largest_component_size}} of Sec. \ref{subsec_LC} in polynomial time via \textbf{Algorithm \ref{algo_basis}} due to sparsity and mutual orthogonality of columns in $X$.

\begin{algorithm}
\caption{Sparsest basis search algorithm for \textnormal{null}$(L_R)$}
\label{algo_basis}
\begin{algorithmic}[1]
  \State Obtain a linearly independent basis $Y$ for \textnormal{null}$(L_R)$.
  \State Compute the number of nonzero and distinct nonzero entries for each column vector $y_i$ in $Y$.
  \State Select the sparsest column vector of $Y$. If there are more than one such vectors then choose the vector with the most distinct entries.
  \State Decompose the chosen vector according to its nonzero distinct entries. For each distinct entry, let $e$ be the binary vector such that its nonzero element is at the same location as the chosen entry. If $e$ is orthogonal to the column vectors in $X$, then add $e$ to $X$.
  \State Repeat step $3)$ and $4)$ until \textnormal{rank}$(X)$ $=$ \textnormal{rank}$(Y)$.
\end{algorithmic}
\end{algorithm}

Note that singular value decomposition (SVD) or QR decomposition methods can be used to find a matrix $Y$ whose column vectors are a basis of \textnormal{null}$(L_R)$.
Since each column vector of $y$ can be represented as the linear combination of the column vectors of $X$ and there is exactly one nonzero entry in each row of $X$, the number of distinct entries of $y_i$ is the number of active column vectors in $X$ that contribute to $y_i$. In addition, due to sparsity and mutual orthogonality of columns in $X$, the greedy basis search can be employed by selecting the sparsest column vector from $Y$ and decompose the vector into several binary vectors and verify the mutual orthogonality property. The criterion in \cite{Coleman86} guarantees that this basis search approach terminates in a finite number of steps since the $\{x_i\}_{i=1}^{n-r}$ are of finite dimension and the result leads to the matrix $X$ in \textbf{Theorem \ref{Thm_largest_component_size}}.

For illustration, consider a network with four nodes, where there is only one edge between node $1$ and node $2$. The graph Laplacian matrix is
$L=\left[
  \begin{smallmatrix*}[c]
    1 &-1 &0 &0 \\
    -1& 1 &0 &0 \\
    0 & 0 &0 &0 \\
    0 & 0 &0 &0
  \end{smallmatrix*} \right]$.
The matrix of our interest is
$X=\left[
  \begin{smallmatrix}
    1  &0 &0 \\
    1  &0 &0 \\
    0  &1 &0 \\
    0  &0 &1
  \end{smallmatrix} \right]$,
and the matrix we obtain from SVD is
$Y=\left[
  \begin{smallmatrix*}[c]
    0.5 &0.5  &0  \\
    0.5& 0.5 &0  \\
    0.5 & -0.5 &1 / \sqrt{2} \\
    0.5 & -0.5 &-1 / \sqrt{2}
  \end{smallmatrix*} \right]$.
Following the aforementioned procedures for reconstructing $X$ from $Y$, the number of nonzero entries for $y_i$ is $4$, $4$ and $2$, and the number of distinct nonzero entries is $1$, $2$ and $2$, respectively. Therefore we start from $y_3$ and decompose it into two vectors $[0~0~1~0]^T$ and $[0~0~0~1]^T$. We add these two vectors to $X$ since they are orthogonal to each other. Then we decompose $y_2$ into $[0~0~1~1]^T$ and $[1~1~0~0]^T$. Since $[0~0~1~1]^T$ is not orthogonal to the vectors in $X$, we discard this vector. Finally, the vector $[1~1~0~0]^T$ is added into $X$ by ckecking the orthogonality property and we obtain the matrix $X$ of interest.

To sum up, with the aid of \textbf{Theorem \ref{Thm_link}} and \textbf{Theorem \ref{Thm_largest_component_size}}, the node removal problem in (\ref{eqn_general_attack}) can be reformulated as
\begin{align}
\label{eqn_attack_Laplacian}
R^*=\min_{R \in F_q} \lambda_1(\widetilde{Q}_R),
\end{align}
where
\begin{align}
\label{eqn_Fq}
F_q=\{X: L_R X=\underline{0},~|R|=q,~X=\arg \min_{X^\prime}\|X^\prime\|_1 \}.
\end{align}
In other words, finding the most disruptive node removal set when removing $q$ nodes from the network is equivalent to solving the minimum matrix one-norm problem on $X$ and then minimizing the largest eigenvalue of $\widetilde{Q}_R$.

\section{Greedy Node Removal Algorithm}
\label{sec_greedy}
It remains to specify a node removal strategy that achieves the minimum in (\ref{eqn_attack_Laplacian}). We propose a node removal strategy to reduce computational complexity, i.e., a greedy node removal algorithm based on spectral graph cut to successively remove the most vulnerable single node. In other words, we recursively solve the $q=1$ version of (\ref{eqn_attack_Laplacian}) until the desired number of nodes have been removed.

\begin{algorithm}
\caption{Greedy node removal based on spectral cut}
\label{algo_threshold}
\begin{algorithmic}[1]
\State \textbf{Input:} $G$ and $|R|=q$
\State \textbf{Output:} $R$
\State $R=\emptyset$
\For{$i=1$ to $q$}
    \State Compute $\widehat{s}$ and $V^{cut}$ in the largest component.
    \State Solve $F_q=\{v^*:~v^*=\arg_{v \in V^{cut}} \|X\|_1$\}
        \If{($|F_q|=1$)}
            \State Set $R= R \cup v^*$
        \Else
            \State $u^*=\arg \min_{u \in F_q} \lambda_1(\widetilde{Q}_R)$
            \State Set $R= R \cup u^*$
        \EndIf
\EndFor
\end{algorithmic}
\end{algorithm}

The spectral cut is associated with the second smallest eigenvector of $L$ (also known as the Fiedler vector \cite{Fiedler73}). For a connected network $G$, let $s$ be the identification vector such that $s_i=1$ if $i$-th node is in group $1$ and $s_i=-1$ if $i$-th node is in group $2$.
The cut size is the number of edges between these two groups, where
\begin{align}
\label{eqn_cut_size}
\textnormal{cut~size}=\frac{1}{4}\sum_{i,j}A_{ij}(1-s_is_j)=\frac{1}{4}s^T L s.
\end{align}
By relaxing $s$ to be real valued and using the fact that $L\textbf{1}=0$, the graph partition problem is equivalent to finding an eigenvector of $L$ that is orthogonal to $\mathbf 1$ such that $s^T L s$ is minimized \cite{Luxburg07}. This is an easily computable approximation to the NP-hard graph partitioning problem. We have
\begin{align}
\label{eqn_Fieder}
s^*=\arg \min_{s \bot \textbf{1}} s^T L s,
\end{align}
where $s^*$ is an eigenvector of the second smallest eigenvalue of $L$.
The partitioning vector is $\widehat{s}=\textnormal{sgn}(s^*)$, where $\textnormal{sgn}(s_i)=1$ if $s_i>0$ and otherwise $\textnormal{sgn}(s_i)=-1$. Define the spectral cut as the set $\{(i,j)\in E:~\widehat{s}_i \widehat{s}_j=-1,~A_{ij}=1\}$, and denote $V^{cut}$ the set of nodes incident to the spectral cut, to be the set of candidate nodes for removal. The optimization in (\ref{eqn_attack_Laplacian}), with $q=1$, is then restricted to the set of nodes in $V^{cut}$, which is a much smaller set than the entire set $V$ of nodes in the graph. This spectral cut and minimization of (\ref{eqn_attack_Laplacian}) process is repeated $q$ times resulting in a significant reduction in computational complexity.   \textbf{Algorithm \ref{algo_threshold}} summarizes the greedy node removal procedure.

\section{Performance Evaluation}
\label{sec_performance}

 An empirical dataset collected from the western states' power subgrid in the United States \cite{Watts98} is used to evaluate the proposed greedy spectral method of node removal. In this network, the nodes represent power stations and the edges represent transmission lines or transformers. The spectral cut based greedy algorithm is compared with node degree and betweenness based greedy node removals proposed in \cite{Holme02}.

\begin{figure}[t]
    \centering
    \includegraphics[width=2.8in]{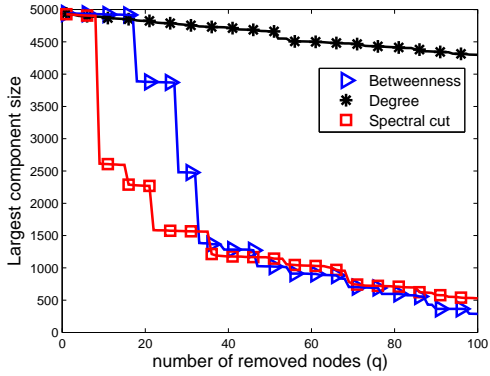}
    \caption{Performance evaluation of greedy node removals based on different node centralities in western states power grid of the United States. This network contains 4941 nodes and 6594 edges. The proposed greedy spectral cut method better reduces the largest component size in the network than do methods based on minimizing degree or betweenness centrality.}
    \label{Fig_USgrid_full}
\end{figure}
\begin{figure}[t]
    \centering
    \includegraphics[width=2.8in]{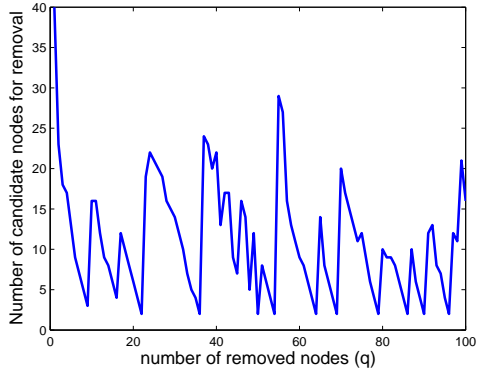}
    \caption{The number of candidate nodes for removal $V^{cut}$ determined by our greedy spectral cut method. In the first five iterations of our greedy algorithm, the average cardinality of $V^{cut}$ is less than 1\% of the total number of nodes (4941) in the network. This proportion decreases to less than 5\% of the nodes as $q$ increases.}
    \label{Fig_candidate_set}
\end{figure}

%\subsection{US Power Grid}
Fig. \ref{Fig_USgrid_full} displays the node removal performances in US power grid, and
it is quite surprising that the largest component size drastically reduces to half of its original size by simply removing 10 nodes from the network using the proposed spectral cut method, whereas we need to remove 30 nodes to achieve the same performance if one uses betweenness measure.
Degree based node removal is not effective in reducing the largest component size, which suggests that although removing nodes with the highest degree seems to be quite intuitive, the high-degree nodes do not necessarily play a key role in topological vulnerability.

In addition, despite the fact that betweenness is a widely adopted measure for evaluating node centrality, it can not fully identify the most vulnerable nodes whose removal maximally reduces the largest component size.
The number of candidate nodes for removal in each iteration are depicted in Fig. \ref{Fig_candidate_set}.
Observe that the number of candidate removal nodes is much smaller than the network size, which makes the proposed greedy node removal strategy effective for large-scale networks and facilitates the assessment of network vulnerability.

\section{conclusion}
\label{sec_con}
Using spectral theory and graph Laplacians, we derive an upper bound on the number of edges in the largest component and we prove that the largest component size minimization problem is equivalent to finding a set of the sparsest orthogonal basis for the null space of the associated graph Laplacian matrix. This basis can be easily constructed using a greedy basis search algorithm with polynomial computational complexity. Experiments on the US power grid dataset show that the proposed greedy node removal algorithm outperforms other approaches based on node degree and betweenness. Our proposed procedure is scalable to large networks and can be used to reveal the vulnerability of modern networks. The method can naturally be applied to exploring the vulnerability of other networks such as biological networks, social networks, and communication networks.

\bibliographystyle{IEEEtran}
\bibliography{IEEEabrv,Laplacian}

\end{document}